\documentclass[conference]{IEEEtran}%
\usepackage{amsmath}
\usepackage{graphicx}
\usepackage{latexsym}
\usepackage{amssymb}
\usepackage{cite}
\usepackage{color}
\usepackage{algorithmic,algorithm}
\usepackage{footnote, stfloats}

\allowdisplaybreaks
\newtheorem{theorem}{Theorem}
\newtheorem{lemma}{Lemma}

\newtheorem{proof}{Proof}
\begin{document}

\title{Secure Full-Duplex Device-to-Device Communication}

\author{\IEEEauthorblockN{Muhammad R. A. Khandaker, Christos Masouros and Kai-Kit Wong}
\IEEEauthorblockA{Department of Electronic and Electrical Engineering\\
University College London\\
Gower Street, London, WC1E 7JE, United Kingdom\\
e-mail: $\{\rm m.khandaker, c.masouros, kai\text{-}kit.wong\}@ucl.ac.uk$}}

\pagestyle{headings}
\maketitle \thispagestyle{empty}
%\vspace*{-2.5em}

\begin{abstract}
This paper considers full-duplex (FD) device-to-device (D2D) communications in a downlink MISO cellular system in the presence of multiple eavesdroppers. The D2D pair communicate sharing the same frequency band allocated to the cellular users (CUs). Since the D2D users share the same frequency as the CUs, both the base station (BS) and D2D transmissions interfere each other. In addition, due to limited processing capability, D2D users are susceptible to external attacks. Our aim is to design optimal beamforming and power control mechanism to guarantee secure communication while delivering the required quality-of-service (QoS) for the D2D link. In order to improve security, artificial noise (AN) is transmitted by the BS. We design robust beamforming for secure message as well as the AN in the worst-case sense for minimizing total transmit power with imperfect channel state information (CSI) of all links available at the BS. The problem is strictly non-convex with infinitely many constraints. By discovering the hidden convexity of the problem, we derive a rank-one optimal solution for the power minimization problem.
\end{abstract}

%\begin{keywords}
%Robust, secrecy, SWIPT.
%\end{keywords}

\section{Introduction}
While traditional half-duplex (HD) communication systems operate in either frequency division duplexing or time division duplexing mode due to lack of practical devices that can transmit and receive concurrently using the same time-frequency resources, recent results in full-duplex (FD) communications have opened up new possibilities to double spectral efficiency in next-generation wireless communications \cite{jrnl_alex_miso_fd, conf_fd_relay_gc, fd_antenna_can, conf_eusipco16}. However, the major detrimental element in FD communications is the so called self-interference (SI) generated by the transmitting node's own signals. In the past, the most demotivating picture of FD communications was the fact that the power of SI can be tens of thousands times higher than that of the signal of interest. Thus the SI is solely strong enough to kill the receiver's decodability.

However, with the advent of small-cell enabled heterogeneous networks (HetNets) as well as transmission in the millimeter-wave (mmWave) frequencies, two key candidates for fifth-generation (5G) networks, wireless transmission distance is showing a sharp decreasing tendency. Thus future communication devices will transmit signals at ultra-low power allowing significant reduction in SI power. This makes concurrent transmission and reception by the same node viable in the same frequency band. Thus FD radios make communication possible that was previously deemed impossible.

Recently, FD communication has attracted affluent interest due to the development of sophisticated SI cancellation techniques. In order to facilitate FD communications in practice, several active and passive SI cancellation techniques have been proposed in both analog and digital domains \cite{conf_alex_icc16}. By combining analog and digital cancellation techniques, antenna cancellation approach proposed in \cite{fd_antenna_can} achieves the amount of self-interference cancellation required for practical full-duplex operation. The method offers up to $60$ dB SI cancellation. In \cite{fd_exp}, the authors have shown that the average amount of SI cancellation increases for active cancellation techniques with increasing received SI power. A major improvement in SI cancellation has been demonstrated in \cite{fd_real_time} with up to $73$ dB cancellation in the digital domain for a $10$ MHz OFDM signal. However, the more advanced full-duplex technique proposed in \cite{fd_1ant} uses a single antenna for simultaneous transmit-receive operation allowing $110$ dB SI cancellation.

Due to the ever increasing number of users and launching of data-demanding services, the importance of efficient use of wireless spectrum can not be overstated. Meanwhile, device-to-device (D2D) communication has turned out to be a promising technology in this regard \cite{d2d_ltea, d2d_survey}, for offloading core-network overhead in particular. In cellular systems, allowing users within close vicinity with high signal to interference and noise ratio (SINR) communicate directly with each other may save significant resources. In such scenario, the D2D users do not need to communicate through the base station (BS), the BS only needs to send the necessary control signals to these users. These users can either use unlicensed bands or share the licensed spectrum. Thus D2D communications offer numerous practically appealing benefits over traditional device-to-base station (D2BS) paradigm including higher spectral efficiency, shorter packet delays, and lower energy consumption. While D2D communication paradigm offers enormous benefits, interference management from D2D transmission to cellular users (CUs) and from cellular transmission to D2D links is crucial. Optimal power control and proper resource allocation schemes can guarantee that cellular systems have multi-dimensional performance gain.
Existing works have considered power control for D2D link to limit the interference and it's detrimental effects on overall system performance \cite{d2d_ltea, d2d_int_lim, d2d_sec_rob, d2d_sec_ss}.
%Interference limited area resource allocation method for
%uplink resources is introduced in \cite{d2d_int_lim}. In this method an area
%around the D2D link will be calculated and the resources of the
%users outside this area will be allocated to D2D link to limit
%the interference from cellular communication to D2D link.

While D2D communication is targeted for future small cells deployments, distance between D2D users should be relatively shorter thus allowing lower transmit power for the D2D communication links. This opens up new prospects for D2D FD communications. Considering the recent advancements in FD radio design that can offer up to $110$ dB SI cancelation \cite{fd_1ant}, along with the maximum transmit power of $20$ dBm for D2D communications, D2D is a well-suited candidate to harvest maximum yield from FD radios.

Nonetheless, D2D communications are more vulnerable to security attacks compared with the traditional D2BS counterpart \cite{conf_const_int_isit17, jrnl_2way_eh, jrnl_alex_miso_fd, d2d_sec_rob, d2d_sec_ss, d2d_sec_survey}. The main reason is the limited computational capacity of the D2D nodes to implement traditional cryptographic security measures. In this context, physical layer security is a viable solution for secure D2D communications. However, with limited processing capability and interference temperature constraints to the main cellular users (CUs), D2D systems need to be carefully designed to guarantee information secrecy such that the legitimate users can correctly decode the confidential information, while the eavesdroppers can retrieve almost nothing from their observations \cite{goel_an, jrnl_secrecy, qli_qos, jrnl_secrecy_sinr}. A jointly optimal resource sharing strategy for power control and channel pairing of CUs and D2D links has been proposed in \cite{d2d_sec_rob} for single-antenna communications. The authors in \cite{d2d_sec_ss} designed security aware power control mechanism for spectrum sharing networks. All these existing works on D2D physical layer security considered HD communications.

In this paper, we consider a downlink MISO cellular system in which a multi-antenna BS transmits a secure message to a single-antenna CU in the presence of multiple eavesdroppers and a pair of D2D communicating nodes. To enhance security, the BS transmits artificial noise (AN) signal superimposed with the secret message signal. The D2D pair communicate using full-duplex radios sharing the frequency band allocated to the CU. Since the D2D users share the same frequency as the CU, both the BS and D2D transmissions interfere each other. In particular, the licensed spectrum holder CU suffers from interference due to the unlicensed D2D transmission. Our aim is to design appropriate power control mechanism to protect the CU's secure communication while guaranteeing the required quality-of-service (QoS) for the D2D link. \textit{To the best of our knowledge, no existing work has addressed the problem for FD D2D communications.} We design robust beamforming for secure message as well as the AN in the worst-case sense with imperfect channel state information (CSI) of all links available at the transmitter. Applying semidefinite relaxation (SDR) techniques, we show that there always exists a rank-one optimal solution for the power minimization problem. Simulation results demonstrate the effectiveness of the proposed algorithm.

The rest of this paper is organized as follows. In Section~\ref{sec_sys}, the system model of MISO downlink system in presence of eavesdroppers and full-duplex D2D communicating nodes is introduced. The joint transmit beamformer and AN design algorithm is developed in Section~\ref{sec_algo_rob} with imperfect CSI of all nodes. Section~\ref{sec_sim} shows the simulation results which justify the significance of the proposed algorithms under various scenarios. Conclusions are drawn in Section~\ref{sec_con}.

\section{System Model}\label{sec_sys}
We consider the downlink of a MISO cellular system in which a multi-antenna BS equipped with $N_{\rm T}$ antennas transmits a secure message to a single-antenna CU in the presence of $K$ non-colluding eavesdroppers each having a single receiving antenna and a pair of D2D communicating nodes. The D2D pair communicate using full-duplex radios sharing the frequency band allocated to a CU carefully selected by the BS. Since the D2D users share the same frequency as the CU, BS transmission interferes the D2D users. At the same time, the CU suffers from the interference due to D2D transmission. To limit this interference, the CU may be chosen for sharing frequency based on the interference-limited-area method proposed in \cite{d2d_int_lim, fd_d2d}.
%Areas of $\sf A_1$ and $\sf A_2$ around the D2D users, in which the D2D interference strength is significantly high, is determined under certain interference temperature constraint and sub-channels of users outside these areas will be assigned for D2D communication. The procedure of calculating the areas $\sf A_1$ and $\sf A_2$ and the power control method will be explained later.
The BS performs transmit beamforming to send secret information to the CU creating minimal interference to the D2D nodes.

For ease of exposition, we name the BS, the legitimate CU, and the eavesdroppers {\it Alice}, {\it Bob}, and {\it Eves}, respectively, while the D2D pair are named {\it David} and {\it Dora} for short. In order to confuse Eves more effectively, it is assumed that Alice transmits artificially generated noise signals superimposed onto the message signal.

Let us now denote ${\bf x}_{\rm s}\in \mathbb C^{N_{\rm T} \times 1}$, ${\bf x}_{\rm n}\in \mathbb C^{N_{\rm T} \times 1}$ as the confidential message signal vector and the AN vector, respectively, and ${\bf h}_{\rm b}\in \mathbb C^{N_{\rm T} \times 1}$, ${\bf h}_{\rm d}\in \mathbb C^{N_{\rm T} \times 1}$, ${\bf h}_{\rm r}\in \mathbb C^{N_{\rm T} \times 1}$, and ${\bf h}_{{\rm e},k}\in \mathbb C^{N_{\rm T} \times 1}$ as the conjugated complex channel vectors between Alice and Bob, David, Dora, and  the $k$th Eve, respectively. The complex channel responses from David and Dora to Bob and $k$th Eve are, respectively, given by $g_{\rm bd}\in \mathbb C, ~~ g_{{\rm de},k}\in \mathbb C$ and $g_{\rm br}\in \mathbb C, ~~ g_{{\rm re},k}\in \mathbb C$. In the aforementioned system model, the received signals at Bob, David, Dora, and $k$th Eve are, respectively, given by
\begin{align}
y_{\rm b}&={\bf h}_{\rm b}^H\left({\bf x}_{\rm s} + {\bf x}_{\rm n}\right) + g_{\rm bd}x_{\rm d} + g_{\rm br}x_{\rm r} + n_{\rm b},\\
y_{\rm d}&={\bf h}_{\rm d}^H\left({\bf x}_{\rm s} + {\bf x}_{\rm n}\right) + g_{\rm dd}x_{\rm d} + g_{\rm dr}x_{\rm r} + n_{\rm d},\\
y_{\rm r}&={\bf h}_{\rm r}^H\left({\bf x}_{\rm s} + {\bf x}_{\rm n}\right) + g_{\rm dr}x_{\rm d} + g_{\rm rr}x_{\rm r} + n_{\rm r},\\
y_{{\rm e},k}&={\bf h}_{{\rm e},k}^H\left({\bf x}_{\rm s} + {\bf x}_{\rm n}\right) + g_{{\rm de},k}x_d + g_{{\rm re},k}x_{\rm r} + n_{{\rm e},k}, \nonumber\\
&\qquad\qquad\qquad~\mbox{for }k=1,\dots, K,
\end{align}
where $g_{\rm dd}\in \mathbb C$ and $g_{\rm rr}\in \mathbb C$ are the self-interfering loop-back channel responses between the transmit and receiving circuitry at David and Dora, respectively, $x_{\rm d}\sim\mathcal{CN}(0,p_{\rm d})$, $x_{\rm r}\sim\mathcal{CN}(0,p_{\rm r})$ are the transmit signals of David and Dora, respectively, $n_{\rm b}\sim\mathcal{CN}(0,\sigma_{\rm b}^2)$, $n_{\rm d}\sim\mathcal{CN}(0,\sigma_{\rm d}^2)$, $n_{\rm r}\sim\mathcal{CN}(0,\sigma_{\rm r}^2)$, and $n_{{\rm e},k}\sim\mathcal{CN}(0,\sigma_{{\rm e},k}^2)$ are the additive Gaussian noises at Bob, David, Dora, and $k$th Eve, respectively. We also assume that the information and AN signals follow distribution ${\bf x}_{\rm s} \sim \mathcal{CN} ({\bf 0},{\bf Q}_{\rm s})$ and ${\bf x}_{\rm n} \sim \mathcal{CN} ({\bf 0},{\bf Q}_{\rm n})$, respectively, where ${\bf Q}_{\rm s}$ is the $N_{\rm T} \times {N_{\rm T}}$ transmit covariance matrix and ${\bf Q}_{\rm n}$ is the $N_{\rm T} \times {N_{\rm T}}$ AN covariance matrix.

Based on the above signal model, the SINR at Bob, David, Dora, and $k$th Eve are, respectively, given by
\begin{align}
{\Gamma}_{\rm b} &= \frac{{\bf h}_{\rm b}^H{\bf Q}_{\rm s}{\bf h}_{\rm b}}{{\bf h}_{\rm b}^H{\bf Q}_{\rm n}{\bf h}_{\rm b} + p_{\rm d}|g_{\rm bd}|^2 + p_{\rm r}|g_{\rm br}|^2 + \sigma_{\rm b}^2}, \label{gmm_b}\\
{\Gamma}_{\rm d} &= \frac{p_{\rm r}|g_{\rm dr}|^2}{{\bf h}_{\rm d}^H\left({\bf Q}_{\rm s}+{\bf Q}_{\rm n}\right){\bf h}_{\rm d} + \rho_{\rm d} +  \sigma_{\rm d}^2},\label{gmm_d}\\
{\Gamma}_{\rm r} &= \frac{p_{\rm d}|g_{\rm dr}|^2}{{\bf h}_{\rm r}^H\left({\bf Q}_{\rm s}+{\bf Q}_{\rm n}\right){\bf h}_{\rm r} + \rho_{\rm r} + \sigma_{\rm r}^2},\label{gmm_r}\\
{\Gamma}_{{\rm e},k} &= \frac{{\bf h}_{{\rm e},k}^H{\bf Q}_{\rm s}{\bf h}_{{\rm e},k}}{{\bf h}_{{\rm e},k}^H{\bf Q}_{\rm n}{\bf h}_{{\rm e},k} + p_{\rm d}|g_{{\rm de},k}|^2 + p_{\rm r}|g_{{\rm re},k}|^2 + \sigma_{{\rm e},k}^2}, \forall k. \label{gmm_ek}
\end{align}
Note that $\rho_{\rm d}$ and $\rho_{\rm r}$ in the SINR expressions for David and Dora in \eqref{gmm_d} and \eqref{gmm_r}, respectively, indicate the residual self-interference, which can be written as $\rho_{\rm d} \sim\mathcal{CN}(0,\sigma_{\rm s,x}^2), ~ {\rm x} \in \{\rm d, r\}$ and $\sigma_{\rm s,x}^2$ can be defined as $\sigma_{\rm s,x}^2 = \alpha_{\rm x}p_{\rm x}$, where $\alpha_{\rm x}$ depends on the amount of SI cancellation at node ${\rm x}$.

In the following, we aim at developing secure precoding schemes under power and interference control criteria in order to ensure maximum network throughput in the D2D communication framework.

\section{Robust Secrecy Beamforming}\label{sec_algo_rob}
In recent studies, it has been demonstrated that the presence of D2D communication links increases overall system throughput if appropriate resource allocation and power control mechanisms are applied \cite{d2d_int_lim, fd_d2d}.

While most of the existing works on secure D2D communications assume perfect CSI, the assumption is not always practical due to the time-varying nature of wireless communication channels. In many practical scenarios, it is often difficult to obtain any information about the eavesdroppers' CSI, or it may even be impractical to assume that an eavesdropper is present at all. Hence in this section, we develop a robust power minimization algorithm considering the worst-case design. In particular, we assume that the actual channels ${\bf h}_{\rm b}$, ${\bf h}_{\rm d}$, ${\bf h}_{\rm r}$, and ${\bf h}_{{\rm e},k}$ lie in the neighbourhood of the estimated channels $\hat{{\bf h}}_{\rm b}$, $\hat{{\bf h}}_{\rm d}$, $\hat{{\bf h}}_{\rm r}$, and $\hat{\bf h}_{{\rm e},k}$, respectively, available at the BS. Hence, the actual channels are modeled as
\begin{eqnarray}
{\bf h}_{\rm b} \!\!\!&=&\!\!\! \hat{{\bf h}}_{\rm b} + {\boldsymbol\delta}_{\rm b}, ~~ 
{\bf h}_{\rm d} = \hat{{\bf h}}_{\rm d} + {\boldsymbol\delta}_{\rm d}\\
{\bf h}_{\rm r} \!\!\!&=&\!\!\! \hat{{\bf h}}_{\rm r} + {\boldsymbol\delta}_{\rm r}, ~~
{\bf h}_{{\rm e},k} = \hat{\bf h}_{{\rm e},k} + {\boldsymbol\delta}_{{\rm e},k}, \forall k,
\end{eqnarray}
in which ${\boldsymbol\delta}_{\rm b}$, ${\boldsymbol\delta}_{\rm d}$, ${\boldsymbol\delta}_{\rm r}$, and ${\boldsymbol\delta}_{{\rm e},k}, \forall k,$ represent the channel uncertainties, which are assumed to be bounded such that
\begin{align}
\|{\boldsymbol\delta}_{\rm b}\|_2 &= \|{\bf h}_{\rm b} - \hat{{\bf h}}_{\rm b}\|_2 \leq \varepsilon_{\rm b}, \mbox{for some }\varepsilon_{\rm b} \geq 0,\\
\|{\boldsymbol\delta}_{\rm d}\|_2 &= \|{\bf h}_{\rm d} - \hat{{\bf h}}_{\rm d}\|_2 \leq \varepsilon_{\rm d}, \mbox{for some }\varepsilon_{\rm d} \geq 0,\\
\|{\boldsymbol\delta}_{\rm r}\|_2 &= \|{\bf h}_{\rm r} - \hat{{\bf h}}_{\rm r}\|_2 \leq \varepsilon_{\rm r}, \mbox{for some }\varepsilon_{\rm r} \geq 0,\\
\|{\boldsymbol\delta}_{{\rm e},k}\|_2 &= \|{\bf h}_{{\rm e},k} - \hat{\bf h}_{{\rm e},k}\|_2 \leq \varepsilon_{{\rm e},k}, \mbox{for some }\varepsilon_{{\rm e},k} \geq 0.
\end{align}
Similarly, the channel coefficients involving the D2D nodes are defined as
\begin{align}
g_{\rm bd} & = \hat g_{\rm bd} + \delta_{\rm bd} ~~\mbox{with}~~ |\delta_{\rm bd}| \le \varepsilon_{\rm bd},\\
g_{\rm br} & = \hat g_{\rm br} + \delta_{\rm br} ~~\mbox{with}~~ |\delta_{\rm br}| \le \varepsilon_{\rm br},\\
g_{\rm dr} & = \hat g_{\rm dr} + \delta_{\rm dr} ~~\mbox{with}~~ |\delta_{\rm dr}| \le \varepsilon_{\rm dr},\\
g_{{\rm de},k} & = \hat g_{{\rm de},k} + \delta_{{\rm de},k} ~~\mbox{with}~~ |\delta_{{\rm de},k}| \le \varepsilon_{{\rm de},k}.
\end{align}
As such, the robust power minimization problem under QoS guarantee is formulated as
\begin{subequations}\label{mxP2}
\begin{eqnarray}
\min_{{\bf Q}_{\rm s},{\bf Q}_{\rm n}\succeq {\bf 0}\atop p_{\rm d}, p_{\rm r} \ge 0} \!\!\!& &\!\!\! {\rm tr}\left({\bf Q}_{\rm s} + {\bf Q}_{\rm n}\right) \label{mxP2_o}\\
{\rm s.t.} \!\!\!& &\!\!\! \min_{\|{\boldsymbol\delta}_{\rm b}\|\leq\varepsilon_{\rm b}} ~~ {\Gamma}_{\rm b} \geq \gamma_{\rm b} \label{mxP2_c0}\\
\!\!\!& &\!\!\! \min_{\|{\boldsymbol\delta}_{\rm d}\|\leq\varepsilon_{\rm d}} ~~ {\Gamma}_{\rm d} \geq \gamma_{\rm d} \label{mxP2_c1}\\
\!\!\!& &\!\!\! \min_{\|{\boldsymbol\delta}_{\rm r}\|\leq\varepsilon_{\rm r}} ~~ {\Gamma}_{\rm r} \geq \gamma_{\rm r} \label{mxP2_c2}\\
\!\!\!& &\!\!\! \max_{\|{\boldsymbol\delta}_{{\rm e},k}\|\leq\varepsilon_{{\rm e},k}} ~~ {\Gamma}_{{\rm e},k} \leq \gamma_{{\rm e},k}, \forall k, \label{mxP2_c3}\\
\!\!\!& &\!\!\! 0 \le p_{\rm d}, p_{\rm r} \le P_{\rm max}.\label{mxP2_c4}
\end{eqnarray}
\end{subequations}
Note that the constraint \eqref{mxP2_c0} guarantees interference control to Bob from the D2D link and the constraint \eqref{mxP2_c1}--\eqref{mxP2_c2} controls interference from the BS transmission to the D2D link, while the constraint \eqref{mxP2_c4} implements power control for the D2D link with $P_{\rm max}$ indicating the maximum allowable transmit power for the D2D nodes. In constraints \eqref{mxP2_c0}--\eqref{mxP2_c3}, there are infinitely many inequalities due to the channel uncertainties, which make the worst-case design particularly challenging. It is guaranteed by \eqref{mxP2_c0}--\eqref{mxP2_c3} that the constraints are satisfied for all realizations of the channel error terms related to the BS to users as well as D2D channels. As such, statistical information about the channel error vectors is not required in this approach, and the knowledge of the upper-bound of the error norms is sufficient.

\textbf{Remark:} Note that when the optimal solution of problem \eqref{mxP2} satisfies the condition ${\rm rank}\big({\bf Q}_{\rm s}\big) \le 1$, transmit beamforming is the optimal strategy for Alice and the transmit power required for the secret message transmission will be minimum. The implementation complexity of corresponding solution is significantly lower. We will prove the existence of a rank-one solution later in this paper. However, we omit the rank constraint for the moment for convenience.

For the worst-case based design in \eqref{mxP2}, the D2D channel gains are upper-(or, lower-)bounded using triangle inequality properties \cite{mat_ana, jrnl_swipt}:
\begin{subequations}\label{tri_ineq}
\begin{align}
|x + y|^2 & \le (|x| + |y|)^2 = |x|^2 + |y|^2 + 2|x|.|y| \label{tri_ineq_pl}\\
|x + y|^2 & \ge (|x| - |y|)^2 = |x|^2 + |y|^2 - 2|x|.|y| \label{tri_ineq_mn}
\end{align}
\end{subequations}
Applying \eqref{tri_ineq}, we have
\begin{subequations}
\begin{align}
\tilde g_{\rm bd} \triangleq \max_{|\delta_{\rm bd}| \le \varepsilon_{\rm bd}} |g_{\rm bd}|^2 &= |\hat g_{\rm bd} + \delta_{\rm bd}|^2 \nonumber\\
\le & |\hat g_{\rm bd}|^2 + \varepsilon_{\rm bd}^2 + 2\varepsilon_{\rm bd}|\hat g_{\rm bd}|\\
\tilde g_{\rm br} \triangleq \max_{|\delta_{\rm br}| \le \varepsilon_{\rm br}} |g_{\rm br}|^2 &= |\hat g_{\rm br} + \delta_{\rm br}|^2 \nonumber\\
\le & |\hat g_{\rm br}|^2 + \varepsilon_{\rm br}^2 + 2\varepsilon_{\rm br}|\hat g_{\rm br}|\\
\tilde g_{\rm dr} \triangleq \min_{|\delta_{\rm dr}| \le \varepsilon_{\rm dr}} |g_{\rm dr}|^2 &= |\hat g_{\rm dr} + \delta_{\rm dr}|^2 \nonumber\\
\le & |\hat g_{\rm dr}|^2 + \varepsilon_{\rm dr}^2 - 2\varepsilon_{\rm dr}|\hat g_{\rm dr}|\\
\tilde g_{{\rm de},k} \triangleq \min_{|\delta_{{\rm de},k}| \le \varepsilon_{{\rm de},k}} & |g_{{\rm de},k}|^2 = |\hat g_{{\rm de},k} + \delta_{{\rm de},k}|^2 \nonumber\\
\le & |\hat g_{{\rm de},k}|^2 + \varepsilon_{{\rm de},k}^2 - 2\varepsilon_{{\rm de},k}|\hat g_{{\rm de},k}|\\
\tilde g_{{\rm re},k} \triangleq \min_{|\delta_{{\rm re},k}| \le \varepsilon_{{\rm re},k}} & |g_{{\rm re},k}|^2 = |\hat g_{{\rm re},k} + \delta_{{\rm re},k}|^2 \nonumber\\
\le & |\hat g_{{\rm re},k}|^2 + \varepsilon_{{\rm re},k}^2 - 2\varepsilon_{{\rm re},k}|\hat g_{{\rm re},k}|.
\end{align}
\end{subequations}
To make \eqref{mxP2} more tractable, we transform the infinitely many inequality constraints \eqref{mxP2_c0}--\eqref{mxP2_c3} into finite linear matrix inequalities (LMIs) by applying ${\mathcal S}$-procedure \cite{boyd}. For completeness, the ${\mathcal S}$-procedure is presented in {\em Lemma \ref{lma1}} below.
\begin{lemma}[\bfseries$\boldsymbol{\mathcal S}$-Procedure] \label{lma1}
Let $f_{i}({\bf x}), i = 1, 2,$ be defined as $$f_{i}({\bf x})={\bf x}^{H}{\bf A}_{i}{\bf x}+2{\Re}\left\{{\bf b}_{i}^{H}{\bf x}\right\}+c_{i}$$
where ${{\bf A}_{i} \in {\mathcal C}^{n\times n}}, {\bf b}_{i} \in {\mathcal C}^{n}, {c}_{i} \in {\mathcal R}$. The implication $f_{1}({\bf x}) \leq 0 \Rightarrow f_{2}({\bf x}) \leq 0$ holds if and only if there exists $\mu \geq 0$ s.t.
$$\mu\left[\begin{array}{cc}{\bf A}_{1}&{\bf b}_{1}\\{\bf b}_{1}^{H}&c_{1}\end{array}\right]-\left[\begin{array}{cc}{\bf A}_{2}&{\bf b}_{2}\cr{\bf b}_{2}^{H}&c_{2}\end{array}\right]\succeq{\bf 0}$$
provided that there exists a point $\hat{{\bf x}}$ such that $f_1(\hat{{\bf x}})<0$.
\end{lemma}

According to ${\mathcal S}$-procedure, if there exists $\mu_{\rm b}  \geq 0$, $\mu_{\rm d}  \geq 0$, $\mu_{\rm r}  \geq 0$, $\mu_{{\rm e},k} \geq 0, \forall k,$ we can transform the constraints \eqref{mxP2_c0}--\eqref{mxP2_c2} into the finite set of LMIs in \eqref{rob_lmis} (top of the next page), where $c_{{\rm e},k} \triangleq \gamma_{{\rm e},k}\left(p_{\rm d}\tilde g_{{\rm de},k} + p_{\rm r}\tilde g_{{\rm re},k} + \sigma_{{\rm e},k}^2\right) - \mu_{{\rm e},k}\varepsilon_k^2$.
%%%%%%%%%%%%%%%%%%%%%%%%%%%%%%%%%%%%%%%%%%%%%%%%%%%%%%%%%%%%%%%%%%%%%%%%%%%%%%%%%%%%%%%%%%%%%%%%%%%%%%%%%%%%
\begin{figure*}
\begin{subequations}\label{rob_lmis}
\begin{eqnarray}
\bar{\boldsymbol\Gamma}_{\rm b}\left({\bf Q}_{\rm s},{\bf Q}_{\rm n},\mu_{\rm b}\right) \!\!\!&\triangleq&\!\!\! \left[\begin{array}{cc}\mu_{\rm b}{\bf I}_{N_{\rm T}} + {\bf Q}_{\rm s} - \gamma_{\rm b}{\bf Q}_{\rm n}&\left({\bf Q}_{\rm s} - \gamma_{\rm b}{\bf Q}_{\rm n}\right)\hat{{\bf h}}_{\rm b}\\
\hat{{\bf h}}_{\rm b}^H\left({\bf Q}_{\rm s} - \gamma_{\rm b}{\bf Q}_{\rm n}\right)&\hat{{\bf h}}_{\rm b}^H\left({\bf Q}_{\rm s} - \gamma_{\rm b}{\bf Q}_{\rm n}\right)\hat{{\bf h}}_{\rm b} - \gamma_{\rm b}\left(p_{\rm d}\tilde g_{\rm bd} + p_{\rm r}\tilde g_{\rm br} + \sigma_{\rm b}^2\right) - \mu_{\rm b}\varepsilon_{\rm b}^2\end{array}\right]\succeq{\bf 0}, \label{sinrB_lmi}\\
\bar{\boldsymbol\Gamma}_{\rm d}\left({\bf Q}_{\rm s},{\bf Q}_{\rm n},\mu_{\rm d}\right) \!\!\!&\triangleq&\!\!\! \left[\begin{array}{cc}\mu_{\rm d}{\bf I}_{N_{\rm T}} - \gamma_{\rm d} \left({\bf Q}_{\rm s} + {\bf Q}_{\rm n}\right) & - \gamma_{\rm d} \left({\bf Q}_{\rm s} + {\bf Q}_{\rm n}\right)\hat{{\bf h}}_{\rm d}\\
- \gamma_{\rm d} \hat{{\bf h}}_{\rm d}^H\left({\bf Q}_{\rm s} + {\bf Q}_{\rm n}\right) & -\gamma_{\rm d} \hat{{\bf h}}_{\rm d}^H\left({\bf Q}_{\rm s} + {\bf Q}_{\rm n}\right)\hat{{\bf h}}_{\rm d} + p_{\rm r}\tilde g_{\rm dr} - \gamma_{\rm d}\left(\rho_{\rm d} + \sigma_{\rm d}^2\right) - \mu_{\rm d}\varepsilon_{\rm d}^2\end{array}\right]\succeq{\bf 0}, \label{sinrD_lmi}\\
\bar{\boldsymbol\Gamma}_{\rm r}\left({\bf Q}_{\rm s},{\bf Q}_{\rm n},\mu_{\rm r}\right) \!\!\!&\triangleq&\!\!\! \left[\begin{array}{cc}\mu_{\rm r}{\bf I}_{N_{\rm T}} - \gamma_{\rm r} \left({\bf Q}_{\rm s} + {\bf Q}_{\rm n}\right) & - \gamma_{\rm r} \left({\bf Q}_{\rm s} + {\bf Q}_{\rm n}\right)\hat{{\bf h}}_{\rm r}\\
- \gamma_{\rm r} \hat{{\bf h}}_{\rm r}^H\left({\bf Q}_{\rm s} + {\bf Q}_{\rm n}\right) & -\gamma_{\rm r} \hat{{\bf h}}_{\rm r}^H\left({\bf Q}_{\rm s} + {\bf Q}_{\rm n}\right)\hat{{\bf h}}_{\rm r} + p_{\rm d}\tilde g_{\rm dr} - \gamma_{\rm r}\left(\rho_{\rm r} + \sigma_{\rm r}^2\right) - \mu_{\rm r}\varepsilon_{\rm r}^2\end{array}\right]\succeq{\bf 0}, \label{sinrR_lmi}\\
\bar{\boldsymbol\Gamma}_{{\rm e},k}\left({\bf Q}_{\rm s},{\bf Q}_{\rm n},\mu_{{\rm e},k}\right) \!\!\!&\triangleq& \!\!\! \left[ \begin{array}{cc}\mu_{{\rm e},k}{\bf I}_{N_{\rm T}} + \gamma_{{\rm e},k} {\bf Q}_{\rm n} - {\bf Q}_{\rm s}&\left(\gamma_{{\rm e},k}{\bf Q}_{\rm n}-{\bf Q}_{\rm s}\right)\hat{{\bf h}}_{{\rm e},k}\\
 \hat{{\bf h}}_{{\rm e},k}^H\left(\gamma_{{\rm e},k}{\bf Q}_{\rm n}-{\bf Q}_{\rm s}\right) &  \hat{{\bf h}}_{{\rm e},k}^H\left(\gamma_{{\rm e},k}{\bf Q}_{\rm n} - {\bf Q}_{\rm s}\right)\hat{{\bf h}}_{{\rm e},k} + c_{{\rm e},k}\end{array} \right]  \succeq {\bf 0}, \label{sinrE_lmi}
\end{eqnarray}
\end{subequations}
\hrulefill \normalsize
\end{figure*}
%\setcounter{equation}{14}
%%%%%%%%%%%%%%%%%%%%%%%%%%%%%%%%%%%%%%%%%%%%%%%%%%%%%%%%%%%%%%%%%%%%%%%%%%%%%%%%%%%%%%%%%%%%%%%%%%%%%%%%%%%%
Substituting the above results into problem \eqref{mxP2}, we can equivalently reformulate the problem as
\begin{subequations}\label{mxP3}
\begin{eqnarray}
\min_{{\bf Q}_{\rm s},{\bf Q}_{\rm n}\succeq {\bf 0}, p_{\rm d}, p_{\rm r} \atop \mu_{\rm b},\mu_{\rm d},\mu_{\rm r},\{\mu_{{\rm e},k}\}} \!\!\!& &\!\!\! {\rm tr}\left({\bf Q}_{\rm s} + {\bf Q}_{\rm n}\right) \label{mxP3_o}\\
{\rm s.t.} \!\!\!& &\!\!\! \bar{\boldsymbol\Gamma}_{\rm b}\left({\bf Q}_{\rm s},{\bf Q}_{\rm n},\mu_{\rm b}\right) \succeq {\bf 0}\label{mxP3_c1}\\
\!\!\!& &\!\!\! \bar{\boldsymbol\Gamma}_{\rm d}\left({\bf Q}_{\rm s},{\bf Q}_{\rm n},\mu_{\rm d}\right) \succeq {\bf 0}\label{mxP3_c2}\\\!\!\!& &\!\!\! \bar{\boldsymbol\Gamma}_{\rm r}\left({\bf Q}_{\rm s},{\bf Q}_{\rm n},\mu_{\rm r}\right) \succeq {\bf 0}\label{mxP3_c3}\\
\!\!\!& &\!\!\! \bar{\boldsymbol\Gamma}_{{\rm e},k}\left({\bf Q}_{\rm s},{\bf Q}_{\rm n},\mu_{{\rm e},k}\right) \succeq{\bf 0}, \forall k,\label{mxP3_c4}\\
\!\!\!& &\!\!\! 0 \le p_{\rm d}, p_{\rm r} \le P_{\rm max}. \label{mxP3_c5}
\end{eqnarray}
\end{subequations}
Problem \eqref{mxP3} is a standard SDP problem, which is convex, and can be optimally solved via off-the-shelf interior-point based solvers \cite{cvx}. Apparently, the solution to the problem \eqref{mxP3} may not in general seem to be optimal to the original problem \eqref{mxP2} due to the rank relaxation considered earlier. However, by discovering the hidden convexity in the problem \eqref{mxP3}, we guarantee the optimality of the proposed solution as described in the following theorem.

\begin{theorem}\label{thm_rank}
Suppose that the SDP problem \eqref{mxP3} is feasible. There always exists an optimal solution $({\bf Q}_{\rm s},{\bf Q}_{\rm n})$ to the problem \eqref{mxP3} such that ${\rm rank}\big({\bf Q}_{\rm s}\big) = 1$.
\end{theorem}

\begin{proof}
See Appendix~\ref{proof_rank_thm}.
\end{proof}

\section{Simulation Results}\label{sec_sim}
In this section, we evaluate the secrecy performance of the proposed robust transmission design through nummerical simulations. For the case of imperfect CSI, the error vectors were uniformly and randomly generated in a sphere centered at zero with the radius $\varepsilon_{\rm b} = \varepsilon_{\rm d} = \varepsilon_{\rm r} = \varepsilon_{{\rm e},k} = \varepsilon = 0.1$. For simplicity, we assume $\gamma_{{\rm e},k} = \gamma_{{\rm e}} , \forall k, \gamma_{{\rm d}} = \gamma_{{\rm r}} = \gamma_{{\rm D2D}}$ unless explicitly mentioned. As a benchmark, we compare the results with perfect CSI cases. All simulation results were averaged over $1000$ random channel realizations.

\begin{figure}[ht]
\centering
\includegraphics*[width=\columnwidth]{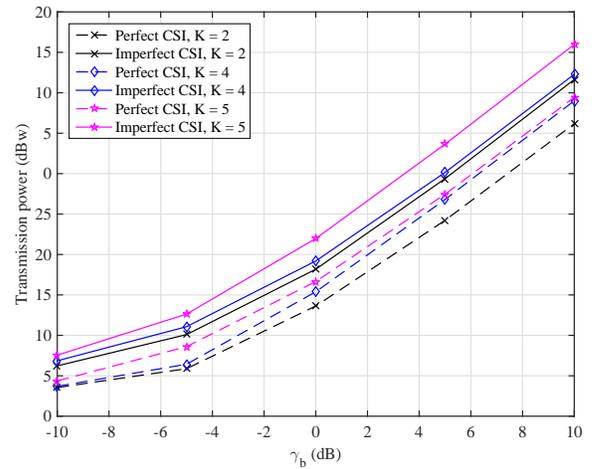}
\caption{Total transmit power $P_{\rm T}$ (dBw) versus Bob's SINR threshold $\gamma_{\rm b}$ with $N_{\rm T} = 10$, $K = 4$, $\gamma_{\rm e} = -5$ (dB), and $\gamma_{{\rm D2D}} = -5$ (dB).}\label{fig_Pt_Gm_b}
\end{figure}

Fig.~\ref{fig_Pt_Gm_b} studies the relation between transmit power and Bob's receive SINR threshold $\gamma_{{\rm b}} $. As expected, the transmit power is monotonically increasing with $\gamma_{\rm b}$. In addition, if the number of eavesdroppers $K$ increases, the total transmit power will increase accordingly to tackle the increased decodability of Eves. It is no surprise that the proposed robust design with imperfect CSI requires higher power compared to the perfect CSI counterpart since the robust algorithm requires to satisfy infinitely many constraints.

\begin{figure}[ht]
\centering
\includegraphics*[width=\columnwidth]{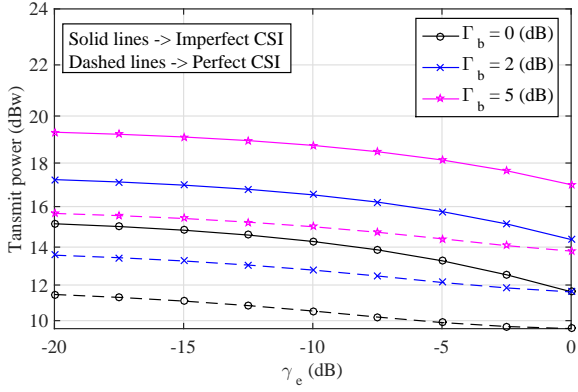}
\caption{Total transmit power $P_{\rm T}$ (dBm) versus Eves' SINR threshold $\gamma_{\rm e}$ with $N_{\rm T} = 10$, $K = 4$, $\gamma_{\rm b} = 10$ (dB), and $\gamma_{{\rm D2D}} = -5$ (dB).}\label{fig_Pt_Gme}
\end{figure}

In the next example, we effects of eavesdropping constraints on BS transmit power. Fig.~\ref{fig_Pt_Gme} shows the transmit power versus Eves' SINR requirement $\gamma_{\rm e}$ with $N_{\rm T} = 10$, $K = 4$, $\gamma_{\rm b} = 10$ (dB), and $\gamma_{{\rm D2D}} = -5$ (dB) for both perfect and imperfect CSI cases. The results in Fig.~\ref{fig_Pt_Gme} show the general trend that with relaxed eavesdropping constraint, smaller transmit power is required. However, increasing $\gamma_{\rm b}$ requires higher transmit power as we observed in Fig.~\ref{fig_Pt_Gm_b} as well.

%\begin{figure}%[h]
%\centering
%\includegraphics*[width=6cm]{Pt_Eh_K_pl.eps}
%\caption{Total transmission power versus harvested power threshold for the case with $N_{\rm T} = 10$, $K = 2, 4, 5$, $\gamma = 0$ (dB), and $\eta = -5$ (dB).}\label{fig_pt_eh}
%\end{figure}

%In the next figure, we analyze the performance of the minimum total transmission power-based design in \eqref{minP1}.
%Results in Fig.~\ref{fig_pt_eh} show the minimum transmit power required for given harvested power constraints as well as satisfying the SINR constraints for $N_{\rm T} = 10$, $K = 2, 4, 5$, $\gamma = 0$ (dB), and $\eta = -5$ (dB)\textcolor{red}{, assuming path loss exponent of $2.7$ which corresponds to an urban cellular network environment}. As can be seen from Fig.~\ref{fig_pt_eh}, with the increasing number of ERs, the required total transmit power increases.

\begin{figure}[ht]
\centering
\includegraphics*[width=\columnwidth]{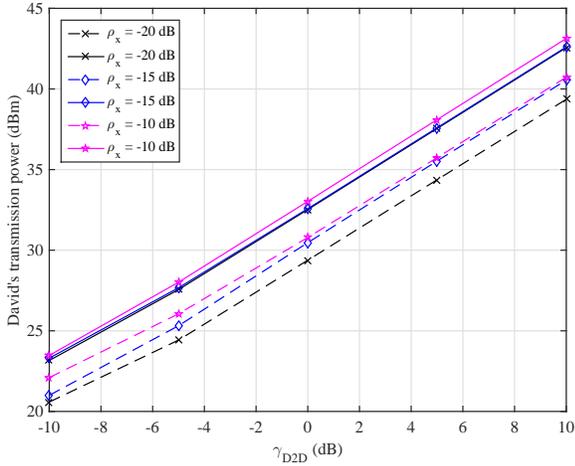}
\caption{David's transmit power $p_{\rm d}$ (dBm) versus SINR threshold of the D2D link $\gamma_{{\rm D2D}}$ with $N_{\rm T} = 10$, $K = 4$, $\gamma_{\rm b} = 10$ (dB), and $\gamma_{\rm e} = -5$ (dB).}\label{fig_region}
\end{figure}

Finally, Fig.~\ref{fig_region} shows the impact of David's transmit power $p_{\rm d}$ on the D2D link link requirement for various residual SI with $N_{\rm T} = 10$, $K = 4$, $\gamma_{\rm b} = 10$ (dB), and $\gamma_{\rm e} = -5$ (dB). It is evident from Fig.~\ref{fig_region} that increased residual SI requires higher power to satisfy the D2D SINR requirement.

\section{Conclusions}\label{sec_con}
This paper studied optimal beamforming and power allocation strategies for secure communication with full-duplex D2D link. We solved the worst-case robust beamforming problem optimally through the discovery of the hidden convexity in the problem. Exploiting FD transmission at all nodes (including BS and Bob) could be an interesting future work.

\appendix
\subsection{Proof of Theorem~\ref{thm_rank}}\label{proof_rank_thm}
Here we use the Karush-Kuhn-Tucker (KKT) conditions of problem \eqref{mxP3} to prove the existence of a rank-one optimal ${\bf Q}_{\rm s}$. Let us now define ${\boldsymbol\Psi}_{\rm s}\succeq{\bf 0}, {\boldsymbol\Psi}_{\rm d}\succeq{\bf 0}, {\boldsymbol\Psi}_{\rm r}\succeq{\bf 0}, {\boldsymbol\Psi}_{{\rm e},k}\succeq{\bf 0}, \forall k, $ as the dual variables associated with the constraints \eqref{mxP3_c1}--\eqref{mxP3_c4}, respectively, whereas $\boldsymbol\Phi_{\rm s} \succeq {\bf 0}$ and $\boldsymbol\Phi_{\rm n} \succeq {\bf 0}$ associated with ${\bf Q}_{\rm s}$ and ${\bf Q}_{\rm n}$, respectively. As such, the Lagrangian of  \eqref{mxP3} can be expressed as
\begin{multline}\label{lagg1}%\mathbb{.} for double-barred letters
\mathcal L\triangleq {\rm tr}\left({\bf Q}_{\rm s} + {\bf Q}_{\rm n}\right) - {\rm tr}\left({\boldsymbol\Psi}_{\rm b}\bar{\boldsymbol\Gamma}_{\rm b}\left({\bf Q}_{\rm s},{\bf Q}_{\rm n},\mu_{\rm b}\right)\right)\\
- {\rm tr}\left({\boldsymbol\Psi}_{\rm d}\bar{\boldsymbol\Gamma}_{\rm d}\left({\bf Q}_{\rm s},{\bf Q}_{\rm n},\mu_{\rm d}\right)\right)
 - {\rm tr}\left({\boldsymbol\Psi}_{\rm r}\bar{\boldsymbol\Gamma}_{\rm r}\left({\bf Q}_{\rm s},{\bf Q}_{\rm n},\mu_{\rm r}\right)\right)\\
- \sum_{k=1}^K {\rm tr}\left({\boldsymbol\Psi}_{{\rm e},k}\bar{\boldsymbol\Gamma}_{{\rm e},k}\left({\bf Q}_{\rm s},{\bf Q}_{\rm n},\mu_{{\rm e},k}\right)\right) - \boldsymbol\Phi_{\rm s}{\bf Q}_{\rm s} - \boldsymbol\Phi_{\rm n}{\bf Q}_{\rm n} + \boldsymbol\Omega,
\end{multline}
where $\boldsymbol\Omega$ includes all the remaining terms not involving ${\bf Q}_{\rm s}$. For ease of exposition, let us now rewrite $\bar{\boldsymbol\Gamma}_{\rm b}$, $\bar{\boldsymbol\Gamma}_{\rm d}$, $\bar{\boldsymbol\Gamma}_{\rm r}$, and $\bar{\boldsymbol\Gamma}_{{\rm e},k}$ as
\begin{subequations}\label{simpli_lmis}
\begin{align}
\bar{\boldsymbol\Gamma}_{\rm b} & ={\boldsymbol\Lambda}_{\rm b}\left(\mu_{\rm b}\right) + \bar{\bf H}_{\rm b}^H\left({\bf Q}_{\rm s} - \gamma_{\rm b}{\bf Q}_{\rm n}\right)\bar{\bf H}_{\rm b},\label{simpli_lmis1}\\
\bar{\boldsymbol\Gamma}_{\rm d} & ={\boldsymbol\Lambda}_{\rm d}\left(\mu_{\rm d}\right) - \gamma_{\rm d}\bar{\bf H}_{\rm d}^H\left({\bf Q}_{\rm s} + {\bf Q}_{\rm n}\right)\bar{\bf H}_{\rm d},\label{simpli_lmis2}\\
\bar{\boldsymbol\Gamma}_{\rm r} & ={\boldsymbol\Lambda}_{\rm r}\left(\mu_{\rm r}\right) - \gamma_{\rm r}\bar{\bf H}_{\rm r}^H\left({\bf Q}_{\rm s} + {\bf Q}_{\rm n}\right)\bar{\bf H}_{\rm r},\label{simpli_lmis3}\\
\bar{\boldsymbol\Gamma}_{{\rm e},k}&={\boldsymbol\Lambda}_{{\rm e},k}\left(\mu_{{\rm e},k}\right) + \bar{\bf H}^H_{{\rm e},k}\left(\gamma_{{\rm e},k}{\bf Q}_{\rm n} - {\bf Q}_{\rm s}\right) \bar{\bf H}_{{\rm e},k},
\end{align}
\end{subequations}
where %{\rm where\,\,\,}
\begin{align}
&{\boldsymbol\Lambda}_{\rm b}\left(\mu_{\rm b}\right) \triangleq \left[\begin{array}{cc}\mu_{\rm b}{\bf I}_{N_{\rm T}}&{\bf 0}\\
{\bf 0} & - \gamma_{\rm b}\left(p_{\rm d}\tilde g_{\rm bd} + p_{\rm r}\tilde g_{\rm br} + \sigma_{\rm b}^2\right) - \mu_{\rm b}\varepsilon_{\rm b}^2\end{array}\right],\nonumber\\
&\bar{\bf H}_{\rm b} \triangleq \left[{\bf I}_{N_{\rm T}}\,\,\,\hat{\bf h}_{\rm b}\right], ~~ \bar{\bf H}_{{\rm e},k} \triangleq \left[{\bf I}_{N_{\rm T}} ~~ \hat{{\bf h}}_{{\rm e},k}\right],\nonumber\\
&{\boldsymbol\Lambda}_{\rm d}\left(\mu_{\rm d}\right) \triangleq \left[\begin{array}{cc}\mu_{\rm d}{\bf I}_{N_{\rm T}}&{\bf 0}\\
{\bf 0} & p_{\rm r}\tilde g_{\rm dr} - \gamma_{\rm d}\left(\rho_{\rm d} + \sigma_{\rm d}^2\right) - \mu_{\rm d}\varepsilon_{\rm d}^2\end{array}\right],\nonumber\\
&\bar{\bf H}_{\rm d} \triangleq \left[{\bf I}_{N_{\rm T}}\,\,\,\hat{\bf h}_{\rm d}\right], ~~ \bar{\bf H}_{{\rm r}} \triangleq \left[{\bf I}_{N_{\rm T}} ~~ \hat{{\bf h}}_{{\rm r}}\right],\nonumber\\
&{\boldsymbol\Lambda}_{\rm r}\left(\mu_{\rm r}\right) \triangleq \left[\begin{array}{cc}\mu_{\rm r}{\bf I}_{N_{\rm T}}&{\bf 0}\\
{\bf 0} & p_{\rm d}\tilde g_{\rm dr} - \gamma_{\rm r}\left(\rho_{\rm r} + \sigma_{\rm r}^2\right) - \mu_{\rm r}\varepsilon_{\rm r}^2\end{array}\right],\nonumber\\
& {\boldsymbol\Lambda}_{{\rm e},k}\left(\mu_{{\rm e},k}\right) \triangleq \left[\begin{array}{cc}\mu_{{\rm e},k}{\bf I}_{N_{\rm T}} & {\bf 0}\\
{\bf 0} & c_{{\rm e},k}\end{array}\right].\nonumber
\end{align}
The relevant KKT conditions can be defined as
\begin{subequations}\label{kkt1}
\begin{eqnarray}
\nabla_{{\bf Q}_{\rm s}}\mathcal L \!\!\!&=&\!\!\! {\bf 0},\label{kkt1_1}\\
\bar{\boldsymbol\Gamma}_{\rm b}\left({\bf Q}_{\rm s},{\bf Q}_{\rm n},\mu_{\rm b}\right){\boldsymbol\Psi}_{\rm b} \!\!\!&=&\!\!\! {\bf 0},\label{kkt1_2}\\
\bar{\boldsymbol\Gamma}_{\rm d}\left({\bf Q}_{\rm s},{\bf Q}_{\rm n},\mu_{\rm d}\right){\boldsymbol\Psi}_{\rm d} \!\!\!&=&\!\!\! {\bf 0},\label{kkt1_3}\\
\bar{\boldsymbol\Gamma}_{\rm r}\left({\bf Q}_{\rm s},{\bf Q}_{\rm n},\mu_{\rm r}\right){\boldsymbol\Psi}_{\rm r} \!\!\!&=&\!\!\! {\bf 0},\label{kkt1_4}\\
\boldsymbol\Phi_{\rm s}{\bf Q}_{\rm s} \!\!\!&=&\!\!\! {\bf 0},\label{kkt1_5}\\
{\boldsymbol\Psi}_{\rm s} \succeq {\bf 0}, {\boldsymbol\Psi}_{\rm d} \succeq {\bf 0}, {\boldsymbol\Psi}_{\rm r} \succeq {\bf 0}, {\boldsymbol\Psi}_{{\rm e},k} \succeq {\bf 0}, \forall k. \label{kkt1_6}
\end{eqnarray}
\end{subequations}
%{\bf A}{\bf Q}_{\rm I} = 0, ~\mbox{and }{\bf D}{\bf Q}_{\rm E} = 0.
Using \eqref{simpli_lmis}, the KKT condition \eqref{kkt1_1} can be expressed as
\begin{multline}
{\bf I}_{N_{\rm T}} - \bar{\bf H}_{\rm b}{\boldsymbol\Psi}_{\rm b}\bar{\bf H}_{\rm b}^H + \gamma_{\rm d}\bar{\bf H}_{\rm d}\bar{\bf H}_{\rm d}^H + \gamma_{\rm r}\bar{\bf H}_{\rm r}\bar{\bf H}_{\rm r}^H + \sum_{k=1}^K\bar{\bf H}_{{\rm e},k}{\boldsymbol\Psi}_{{\rm e},k}\bar{\bf H}^H_{{\rm e},k}\\
= {\boldsymbol\Psi}_{\rm s}.\label{kkt1_12}
\end{multline}
Let $\boldsymbol\Sigma$ denote the positive-definite matrix
\begin{align}
\boldsymbol\Sigma &\triangleq  {\bf I}_{N_{\rm T}} + \gamma_{\rm d}\bar{\bf H}_{\rm d}\bar{\bf H}_{\rm d}^H + \gamma_{\rm r}\bar{\bf H}_{\rm r}\bar{\bf H}_{\rm r}^H + \sum_{k=1}^K\bar{\bf H}_{{\rm e},k}{\boldsymbol\Psi}_{{\rm e},k}\bar{\bf H}^H_{{\rm e},k}.\label{srcBb}
\end{align}
Multiplying both sides of \eqref{kkt1_12} by ${\bf Q}_{\rm s}$ and using KKT condition \eqref{kkt1_5}, we have
\begin{align}
\boldsymbol\Sigma {\bf Q}_{\rm s} = \bar{\bf H}_{\rm b}{\boldsymbol\Psi}_{\rm b}\bar{\bf H}_{\rm b}^H{\bf Q}_{\rm s}
\end{align}
Thus it is obvious that
\begin{align}
{\rm rank}\left({\bf Q}_{\rm s}\right) = {\rm rank}\left(\boldsymbol\Sigma{\bf Q}_{\rm s}\right) \le  {\rm rank}\left(\bar{\bf H}_{\rm b}{\boldsymbol\Psi}_{\rm b}\bar{\bf H}_{\rm b}^H{\bf Q}_{\rm s}\right)\label{rank_Qs}
\end{align}
since pre(post)-multiplying any matrix by a positive-definite matrix does not change its rank. Next, we show that ${\rm rank}(\bar{\bf H}_{\rm b}{\boldsymbol\Psi}_{\rm b}\bar{\bf H}_{\rm b}^H)\leq 1$. Substituting \eqref{simpli_lmis1} into the KKT condition \eqref{kkt1_2}, we obtain
\begin{equation}\label{kkt2_1}
{\boldsymbol\Lambda}_{\rm b}\left(\mu_{\rm b}\right){\boldsymbol\Psi}_{\rm b} + \bar{\bf H}_{\rm b}^H\left({\bf Q}_{\rm s} - \gamma{\bf Q}_{\rm n}\right)\bar{\bf H}_{\rm b}{\boldsymbol\Psi}_{\rm b} = {\bf 0}.
\end{equation}
Post-multiplying \eqref{kkt2_1} by $\bar{\bf H}_{\rm b}^H$ yields
\begin{equation}
{\boldsymbol\Lambda}_{\rm b}\left(\mu_{\rm b}\right){\boldsymbol\Psi}_{\rm b}\bar{\bf H}_{\rm b}^H + \bar{\bf H}_{\rm b}^H\left({\bf Q}_{\rm s} - \gamma{\bf Q}_{\rm n}\right)\bar{\bf H}_{\rm b}{\boldsymbol\Psi}_{\rm b}\bar{\bf H}_{\rm b}^H = {\bf 0}.\label{kkt2_2}
\end{equation}

Now, the following facts can be easily verified:
\begin{align}
\left[{\bf I}_{N_{\rm T}}\,\, {\bf 0}\right]\bar{\bf H}_{\rm b}^H & = {\bf I}_{N_{\rm T}}\notag\\
\left[{\bf I}_{N_{\rm T}}\,\, {\bf 0}\right]{\boldsymbol\Lambda}_{\rm b}\left(\mu_{\rm b}\right) & = \mu_{\rm b}\left[{\bf I}_{N_{\rm T}}\,\, {\bf 0}\right] = \mu_{\rm b}\left(\bar{\bf H}_{\rm b} - \left[{\bf 0}_{N_{\rm T}} \,\, \hat{\bf h}_{\rm b}\right]\right).\notag
\end{align}
Premultiplying both sides of \eqref{kkt2_2} by $\left[{\bf I}_{N_{\rm T}}\,\, {\bf 0}\right]$ we get
%\begin{equation}
%\mu_{\rm I}\left(\bar{\bf H}_{\rm I} - \left[{\bf 0}_{N_{\rm T}} \,\, \hat{\bf h}_{\rm I}\right]\right){\boldsymbol\Psi}_{\rm I}\bar{\bf H}_{\rm I}^H + \left({\bf Q}_{\rm I} - \gamma{\bf Q}_{\rm E}\right)\bar{\bf H}_{\rm I}{\boldsymbol\Psi}_{\rm I}\bar{\bf H}_{\rm I}^H = {\bf 0},\notag
%\end{equation}
%which is equivalent to
\begin{equation}\label{kkt2_4}
\left(\mu_{\rm b}{\bf I}_{N_{\rm T}} + {\bf Q}_{\rm s} - \gamma_{\rm b}{\bf Q}_{\rm n}\right)\bar{\bf H}_{\rm b}{\boldsymbol\Psi}_{\rm b}\bar{\bf H}_{\rm b}^H = \mu_{\rm b}\left[{\bf 0}_{N_{\rm T}} \,\, \hat{\bf h}_{\rm b}\right]{\boldsymbol\Psi}_{\rm b}\bar{\bf H}_{\rm b}^H.
\end{equation}
%Next, we describe the following lemma \cite{mat_ana}:

\begin{lemma}\label{lemm_psd}
If a hermitian matrix ${\bf M} = \left[\begin{array}{cc}{\bf M}_{11} & {\bf M}_{12}\\
{\bf M}_{21} & {\bf M}_{22}
\end{array}\right]\succeq {\bf 0},$ then it immediately follows that ${\bf M}_{11}$ and ${\bf M}_{22}$ must be PSD matrices \cite{mat_ana}.
\end{lemma}

Applying {\it Lemma~\ref{lemm_psd}} to \eqref{sinrB_lmi}, we have that $\mu_{\rm b}{\bf I}_{N_{\rm T}} + {\bf Q}_{\rm s} - \gamma_{\rm b}{\bf Q}_{\rm n} \succeq {\bf 0}$. Then the following rank relation holds
\begin{eqnarray}
{\rm rank}\left(\bar{\bf H}_{\rm b}{\boldsymbol\Psi}_{\rm b}\bar{\bf H}_{\rm b}^H\right) \!\!\!\!\!&=&\!\!\!\!\! {\rm rank}\left(\left(\mu_{\rm b}{\bf I}_{N_{\rm T}} + {\bf Q}_{\rm s} - \gamma_{\rm b}{\bf Q}_{\rm n}\right)\!\bar{\bf H}_{\rm b}{\boldsymbol\Psi}_{\rm b}\bar{\bf H}_{\rm b}^H\right)\nonumber\\
\!\!\!&=&\!\!\! {\rm rank}\left(\mu_{\rm b}\left[{\bf 0}_{N_{\rm T}} \,\, \hat{\bf h}_{\rm b}\right]{\boldsymbol\Psi}_{\rm b}\bar{\bf H}_{\rm b}^H\right)\nonumber\\
\!\!\!&\leq&\!\!\! {\rm rank}\left(\left[{\bf 0}_{N_{\rm T}} \,\, \hat{\bf h}_{\rm b}\right]\right) \leq 1 \label{rank_ineq}
\end{eqnarray}
where \eqref{rank_ineq} follows from a basic rank inequality property \cite{mat_ana}. Note that ${\bf Q}_{\rm s} = {\bf 0}$ does not satisfy the constraint \eqref{mxP2_c0} hence can not be the optimal solution to problem \eqref{mxP3}. Thus combining \eqref{rank_Qs} and \eqref{rank_ineq}, we conclude that ${\rm rank}\left({\bf Q}_{\rm s}\right) = 1$ must hold.\hfill $\Box$

\bibliographystyle{IEEEtran}\footnotesize{
%\IEEEtriggeratref{7}
\bibliography{IEEEabrv,../../../../../../refdb}}

\end{document}